\documentclass[a4paper,UKenglish,cleveref,numberwithinsect,thm-restate]{lipics-v2021}

\pdfoutput=1 

\title{Nearly-Tight Bounds for Flow Sparsifiers in Quasi-Bipartite Graphs}

\author{Syamantak Das}{IIIT Delhi, India}{syamantak@iiitd.ac.in}{https://orcid.org/0000-0002-4393-8678}{}
\author{Nikhil Kumar}{University of Waterloo, Canada }{nikhil.kumar2@uwaterloo.ca}{https://orcid.org/0000-0001-8634-6237}{}
\author{Daniel Vaz}%
    {LAMSADE, CNRS, Université Paris-Dauphine, Université PSL, France}%
    {daniel.ramos-vaz@dauphine.psl.eu}%
    {https://orcid.org/0000-0003-2224-2185}%
    {This work was funded by grants ANR-19-CE48-0016 (Algoridam) and ANR-21-CE48-0022 (S-EX-AP-PE-AL) from the French National Research Agency (ANR), and partially done while the author was at École Normale Superieure and IRIF. }

\authorrunning{S. Das, N. Kumar, and D. Vaz}
\Copyright{Syamantak Das, Nikhil Kumar, and Daniel Vaz}


\category{} 

\nolinenumbers % 

    \relatedversion{Conference version published at MFCS 2024, \href{https://dx.doi.org/10.4230/LIPIcs.MFCS.2024.45}{doi:10.4230/LIPIcs.MFCS.2024.45}}
    \hideLIPIcs

\usepackage{mystyle-lipics}

\DeclareMathOperator{\mc}{mc}

\authorcommands[green!80!black]{daniel}{Daniel}{D}
\authorcommands[red!80!black]{syam}{Syam}{S}

\newcommand{\bfd}{\ensuremath{\mathbf d}\xspace}
\newcommand{\bfdt}{\ensuremath{\mathbf {\tilde d}}\xspace}
\newcommand{\dt}{\ensuremath{{\tilde d}}\xspace}

\begin{document}

\maketitle

\begin{abstract}
Flow sparsification is a classic graph compression technique which, given a capacitated graph $G$ on $k$ terminals, aims to construct another capacitated graph $H$, called a \emph{flow sparsifier}, that preserves, either exactly or approximately, every \emph{multicommodity flow} between terminals (ideally, with size as a small function of $k$). 
Cut sparsifiers are a restricted variant of flow sparsifiers which are only required to preserve maximum flows between bipartitions of the terminal set. 
It is known that exact cut sparsifiers require $2^{\Omega(k)}$ many vertices [Krauthgamer and Rika, SODA 2013],
with the hard instances being \emph{quasi-bipartite} graphs, {where there are no edges between non-terminals}. 
On the other hand, it has been shown recently that exact (or even $(1+\varepsilon)$-approximate) flow sparsifiers on networks with just 6 terminals require unbounded size [Krauthgamer and Mosenzon, SODA 2023, Chen and Tan, SODA 2024]. 

In this paper, we construct exact flow sparsifiers of size $3^{k^{3}}$ and exact cut sparsifiers of size $2^{k^2}$ for quasi-bipartite graphs. 
In particular, the flow sparsifiers are contraction-based, that is, they are obtained from the input graph by (vertex) contraction operations. 
Our main contribution is a new technique to construct sparsifiers that exploits connections to polyhedral geometry, and that can be generalized to graphs with a small separator that separates the graph into small components.
We also give an improved reduction theorem for graphs of bounded treewidth~[Andoni et al., SODA 2011], implying a flow sparsifier of size $\bigO(k\cdot w)$ and quality $\bigO\paren[\big]{\frac{\log w}{\log \log w}}$, where $w$ is the treewidth.
\end{abstract}

\section{Introduction}
Graph sparsification is a classic and influential technique in algorithm design. The idea behind graph sparsification is to compress a given graph into a ``smaller'' graph (the notion of small depends on the context) which preserves certain crucial properties of the graph. 
{The notion of edge sparsification dates back to the work of Gomory and Hu~\cite{GomoryH61} and to Nagamochi and Ibaraki~\cite{NagamochiI92}, who developed techniques to find a sparser graph -- that is, with fewer edges -- preserving $s$-$t$-cuts and $k$-edge-connectivity, respectively. %
This work was continued by Benczur and Karger~\cite{BenczurK00} and Spielman and Teng~\cite{SpielmanT04/stoc}, who extended the techniques to preserving cut values, or more generally the Laplacian spectrum, up to a factor of $1+\varepsilon$.}

Relatively recent is the study of vertex sparsification. Arguably, the most extensively explored notions here are flow sparsification and cut sparsification. 
Specifically, suppose we are given an undirected graph $G=(V,E)$ along with a capacity function $u$ on the edges and a subset of vertices $T$ called terminals, with $|T| = k$. 
A cut sparsifier $H$ of quality $q\geq 1$ is a graph on (potentially) fewer vertices which preserves the minimum cut value between every possible bipartition of the terminal set, up to a factor of $q$. 
A more general notion is flow sparsification where the sparsifier must preserve all multicommodity flows between the terminals (formal definitions are introduced in \Cref{sec:prelim}).
Hence, we can see cut sparsification as a special case where it is only required to preserve the single-commodity flows between bipartitions of terminals. 

The main focus in cut and flow sparsification research is to strike the ideal trade-off between the size of $H$ and its quality $q$. 
In their seminal work, Moitra~\cite{Moitra09} and later {Moitra and} Leighton~\cite{LeightonM10/stoc} showed that there is a flow sparsifier on just the terminal set (that is with size $k$) of quality $\bigO\paren[\big]{\frac{\log k}{\log\log k}}$, and their work was later made constructive by different works~\cite{CharikarLLM10,EnglertGKRTT14,MakarychevM10}. %
These works also showed that any sparsifier of size $k$ would have a quality loss of at least $\Omega(\sqrt{\log k/\log \log k})$~\cite{MakarychevM10}. 
Hence, a significant improvement in the quality would either require more vertices in the sparsifier or special properties of the graph. 

Considerable research effort has been dedicated to cut and flow sparsification in more restricted settings.
For instance, one can construct a flow sparsifier on only the terminals with quality $\bigO(r)$ for graphs that exclude $K_{r,r}$ as a minor~\cite{CalinescuKR04,Moitra09}, and of quality $2$ for unweighted trees and quasi-bipartite graphs~\cite{GoranciR16}; these results exploit connections between flow sparsification and the $0$-extension problem.
For general graphs, Chuzhoy~\cite{Chuzhoy12} gave a construction with quality $O(1)$ and size $C^{\bigO(\log\log C)}$, where $C$ is the total capacity of edges incident on terminals (and hence might be as large is $\Omega(nk)$). 
On the other hand, one can construct quality-1 cut sparsifiers of size $\bigO(2^{2^k})$ for general graphs~\cite{HagerupKNR98,KhanR14/ipl}, of size $\bigO(k^2\cdot 2^k)$ for {planar graphs~\cite{KrauthgamerR13}} and $\bigO(k) \cdot 2^{2^{\bigO(w)}}$ for graphs with treewidth $w$~\cite{ChaudhuriSWZ00} (such sparsifiers are also known as mimicking networks or exact sparsifiers in the literature). 
The scenario is drastically different for exact flow sparsifiers. Recent breakthroughs have ruled out the existence of exact flow sparsifiers~\cite{KrauthgamerM23}, {as well as contraction-based} quality-$(1+\varepsilon)$ flow sparsifiers~\cite{ChenT24a/soda}, with size as a function of $k$, {which is achieved} by demonstrating hard instances on 6 terminal networks. 
{On the other hand, contraction-based flow sparsifiers of quality $1+\epsilon$ and size $2^{\tt{poly}\left(1/\varepsilon \right) }$ exist for} every 5-terminal network~\cite{ChenT24a/soda}.

For the special case of \emph{quasi-bipartite graphs}~\cite{RajagopalanV99}, where non-terminals form an independent set, Andoni et al.~\cite{AndoniGK14} improved the bound of Chuzhoy significantly: they give a quality-$(1+\varepsilon)$ flow sparsifier of size $\poly(k/\varepsilon)$, %
{recently improved to size $k \cdot \poly(\log k, \epsilon^{-1})$~\cite{AbrahamDKKP16,JambulapatiLLS23}.}
The significance of these result lies in the fact that for the simpler case of cut sparsification, these graphs present some of the hardest known instances~\cite{KhanR14/ipl,KrauthgamerR13}, where a quality-1 cut sparsifier requires $2^{\Omega(k)}$ vertices~\cite{KhanR14/ipl,KrauthgamerR13}. %
For unweighted quasi-bipartite graphs, this bound is tight~\cite{GoranciR16}, but no construction was known for general capacities.
Thus, the result of Andoni et al.~\cite{AndoniGK14} raises hope that we can overcome the lower bounds by designing quality-$(1+\varepsilon)$ sparsifiers, even in the flow sparsification setting. %

Andoni et al.\ also prove a second result where they show a generic reduction from graphs of bounded treewidth to general graphs in the following sense: 
They give a construction whereby the existence of any quality-$q(k)$ sparsifier of size $S(k)$ implies a quality-$q(6w)$ sparsifier of size $k^4\cdot S(6w)$ where $w$ is the treewidth of the graph. 

\pagebreak
\subparagraph{Our Results.}

We give the following results on quasi-bipartite graphs and their extensions:
\begin{enumerate}
    \item A cut sparsifier of size $2^{k^2}$ for quasi-bipartite graphs (\Cref{thm:bipartite});
    \item A contraction-based flow sparsifier of size $3^{k^3}$ for quasi-bipartite graphs (\Cref{thm:bipartite-flow});
    \item A cut sparsifier of size $k + 2^{c^2}$ and a flow sparsifier of size $k + 3^{c^3}$ when the graph has a vertex cover of size $c$ (\Cref{thm:vc});
    \item A cut sparsifier of size $kd + 4^{d^3}$ when $G$ has vertex integrity $d$~\cite{BarefootES87,GimaHKKO22}, that is, a separator $X\subseteq V$ such that $|X| + |C| \leq d$ for every component $C$ of $G-X$ (\Cref{thm:vi}).
\end{enumerate}
Note that our result almost matches the lower bound given by Krauthgamer and Rika~\cite{KrauthgamerR13} (up to a polynomial in the exponent).
Further, our result on flow sparsifiers {shows that instances on quasi-bipartite graphs are not hard for flow sparsification, as they admit better bounds than general graphs.} 

Our main contribution lies in developing a novel tool for constructing sparsifiers that is based on connections to polyhedral geometry. %
We show that this technique can be applied to obtain cut and flow sparsifiers, and even when the terminal set separates the graph into small components. %
Furthermore, we show that the size of the sparsifier actually grows with the size of the separator whose removal leaves only small components, %
thus obtaining improved results for bounded vertex cover and vertex integrity, two structural graph parameters that {have recently gained popularity in the parameterized community}~\cite{BodlaenderGP23,FominLMT18,GimaHKKO22,GimaO22,LampisM21,OostveenL23}.
{These have particular relevance when studying problems that are hard for more general parameters, such as treewidth~\cite{FialaGK11} and treedepth~\cite{BodlaenderGP23,GimaHKKO22}; they also allow for stronger meta-theorems~\cite{GimaO22,LampisM21} compared to the classic theorem of  Courcelle for bounded-treewidth graphs~\cite{Courcelle90/iandc}.}

We give an additional result for graphs with treewidth $w$, improving the results of Andoni et al.~\cite{AndoniGK14} and Chaudhuri et al.~\cite{ChaudhuriSWZ00}:
we construct a flow sparsifier of size $k\cdot S(2w)$ and quality $g(2w)$ provided that every $k$-terminal network admits a quality $g(k)$ flow sparsifier with size $S(k)$ (see \Cref{sec:btw}). 
This implies an $\bigO\paren[\big]{\frac{\log w}{\log\log w}}$-quality sparsifier with size $\bigO(k\cdot w)$ for graphs with treewidth $w$ using results from~\cite{EnglertGKRTT14, LeightonM10/stoc}.

    The proofs of lemmas marked with an asterisk (*) are deferred to the appendix.

\subparagraph{Concurrent Work.}
Independently of our work, Chen and Tan designed contraction-based cut sparsifiers of size $k^{O(k^2)}$ for quasi-bipartite graphs~\cite{abs-2407-10852}, as well as quality-$(1+\epsilon)$ cut sparsifiers of size $\poly(k,\epsilon)$ for planar graphs.
Their construction for quasi-bipartite graphs is slightly larger than our cut sparsifier (of size $2^{k^2}$), but smaller than our contraction-based construction ($2^{k^3}$); they do not present any flow sparsifiers.

\subparagraph{Techniques.} %
The main idea behind our construction of cut and flow sparsifiers for bipartite graphs is to consider them as a union
of stars centered on Steiner vertices, which can be handled independently as
they do not share any edges. %
By showing that the number of different ways that stars can participate in
cuts is bounded, we get a sparsifier with the same bound on the size, as
equivalent stars can be contracted together. %
Thus, it is sufficient to show how to put stars into a bounded number of classes. %

Our approach is based on polyhedral theory. We represent each star by a vector in $\Rnn^k$ where each coordinate is the capacity of an edge between the center of the star and a terminal. We show that there are $2^{k^2}$ stars, which we
refer to as \emph{basic stars}, such that any star is the conic combination
of at most $k$ of them. %
Using this idea, we obtain two constructions: the first is to construct $H$
from the terminals and the set of \emph{basic stars} with appropriately
scaled-up capacities; whereas the second (slightly larger) is to simply
contract vertices that are the conic combination of the same set of basic
stars. %
Since the second construction is contraction-based, a consequence of our
results is that optimal algorithm for contraction-based sparsifiers presented
by Khan et al.~\cite{KhanR14/ipl} obtains a sparsifier of size at most
$2^{k^3}$ for bipartite graphs.

The construction of the flow sparsifiers is more involved. 
A first attempt would be to use our result for cut sparsifier along with a result from Andoni et al.~\cite[Theorem 7.1]{AndoniGK14} which roughly implies that if the flow-cut gap for the given graph is $\gamma$, then an exact contraction-based cut sparsifier is also a flow sparsifier of quality $\gamma$. 
Unfortunately, bipartite networks can have a flow-cut gap of $\Omega(\log k)$ and hence this approach fails. 

We rather take a more direct approach: by relying on the above mentioned polyhedral tool to define equivalence classes on the stars, we show that we can contract the stars in each equivalence class into a single one. 
Applying this technique is much more challenging in the case of flow sparsifiers, since one has to ensure that every multicommodity flow between terminals (and not just bipartitions) must be preserved. 

The main technical difficulty is to show that merging two stars preserves the routing of multicommodity flows, and in particular, that if a demand can be routed in the merged star, then it can also be routed in the original stars (if the original stars are in the same equivalence class).
We achieve this by splitting the demand so that each part can be routed in a different star, by a process which iteratively adds demand to one of the two stars, and if that is no longer possible, refines the partition by globally switching demands between the two stars. 
We show that when the process can no longer introduce new demand, then there is a saturated cut in the merged star, and thus all of the demand must be already routed.

\section{Preliminaries}
\label{sec:prelim}

A network $G=(V,E,u)$ is a graph $(V,E)$ with edge capacities $u\colon E(G)
\to \Rbb_{\geq 0}$. %
It is usually associated with a set of \emph{terminals} $K \subseteq V(G)$, whose
size we denote by $k$. %
We refer to vertices in $V(G) \setminus K$ as \emph{non-terminal} or
\emph{Steiner} vertices, %
and say that two networks $G_1$, $G_2$ are \emph{Steiner-disjoint} if $V(G_1)
\cap V(G_2) \subseteq K$.

We consider a cut to be a subset of vertices $X \subset V$, with cut edges
$\delta(X) = {E(X, V-X)}$. %
For convenience, we usually write $u(X)$ to mean $u(\delta(X))$ for any $X
\subseteq V$. %
A cut $X$ \emph{separates} $A \subseteq K$ if $A \subseteq X$ and $K-A
\subseteq V-X$, and it is a min-cut for $A$ if it minimizes the capacity among
all cuts separating $A$. %
We denote by $\mc_G(A)$ the smallest (minimum $|X|$) min-cut (in $G$) that
separates $A$, and $\kappa_G(A)$ its capacity. %
If the network is clear from context, we drop the subscript in $\mc(A)$,
$\kappa(A)$.

\subparagraph{Cut sparsifiers.}
A \emph{cut sparsifier} of quality $q\geq 1$ for a network $G$ with terminals $K$
is a network $H$ such that $K \subseteq V(H)$ and for every subset $A \subset
K$, the capacity of the min-cut separating $A$ is $q$-approximated, that is,
\[\kappa_G(A) \leq \kappa_H(A) \leq q\cdot \kappa_G(A).\]
Unless specified, a cut sparsifier is of quality $1$.

\subparagraph{Flow sparsifiers.}
A \emph{flow sparsifier} of quality $q \geq 1$ for a network $G$ with
terminals $K$ is a network $H$ that $q$-approximately preserves the
multi-commodity flows for any demand. %
We use the formal description in the work of Andoni et al.~\cite{AndoniGK14}. %

We say that a demand $\bfd \in \Rbb_{\geq 0}^{K\times K}$ \emph{is routed} in
$G$ by flow $f\geq 0$ if $\sum_{P \in \Pcal_{s,t}} f_P = d\paren{s,t}$ and $\sum_{P \ni
e} f_P \leq u(e)$, where $f$ is defined over paths and $\Pcal_{s,t}$ is the set of
all $s$-$t$-paths. %
We consider both demands and paths as symmetric, that is, $d(s,t) = d(t,s)$,
but $\Pcal_{s,t} = \Pcal_{t,s}$ so the same flows can satisfy both demands.

The demand polytope~\cite{AndoniGK14} for a network is the set all demands
that can be routed in $G$, $\Dcal(G) = \set{\bfd: \text{$\bfd$ can be routed
in $G$}}$.
Given a demand vector $\bfd$, its \emph{flow factor} is the value
$
\lambda_G(\bfd) = \sup\set{\lambda \geq 0 : \lambda \bfd \in \Dcal(G)}.
$

We formally define \emph{flow sparsifiers} as follows: %
$H$ is a quality-$q$ flow sparsifier for $G$ with terminals $K$ if for all
demand vectors $\bfd$,
\[
\lambda_G(\bfd) \leq \lambda_H(\bfd) \leq q\cdot \lambda_G(\bfd).
\]
In particular for $q=1$, we have that $\bfd \in \Dcal(G)$ if
and only if $\bfd \in \Dcal(H)$.

\subparagraph{Treewidth.}
A tree decomposition $(T,\mathcal B)$ of a graph $G$ is a tree $T$ together
with a collection of subsets of vertices $\mathcal B = \set{B_i}_{i \in V(T)}$
called \emph{bags}, such that:
\begin{itemize}
  \item For each edge $uv \in E(G)$, there is a bag $B_i$ containing both $u$ and $v$.
  \item For each vertex $v \in V(G)$, the collection of bags containing $v$ induces a non-empty subtree of $T$.
\end{itemize}
The \emph{width} of $(T,\mathcal B)$ is $w(T,\mathcal B) = \max_{i \in V(T)} (|B_i| - 1)$. %
The \emph{treewidth} of $G$ is the minimum width achievable by any tree decomposition.

We assume that there are no two identical bags in the decomposition, as
otherwise we can simply contract the edge connecting the corresponding nodes. %
We consider that each edge $uv$ is associated with a single node of $T$,
namely the node closest to the root whose bag contains both $u$ and $v$. %
The collection of edges associated with a node $i \in T$ is denoted by $E_i$,
and the subgraph induced by a bag is $G[i] = (B_i, E_i)$. %
This notation is particularly useful when talking about the graph induced by
collections of bags, and thus for a subset $R \subseteq V(T)$ we write $B(R) =
\bigcup_{i \in R} B_i$, $E(R) = \bigcup_{i \in R} E_i$, and $G[R] =
\paren{B(R), E(R)}$. %
We remark that $G[R]$ and $G[B(R)]$ are subgraphs on the same subset of
vertices but with different sets of edges.

Computing the treewidth of a graph, together with the corresponding tree
decomposition, is an NP-hard problem~\cite{ArnborgCP87}, but can be computed in
time $w^{O(w^3)} \cdot n$~\cite{Bodlaender96} for a graph of treewidth $w$. %
For a faster running time, we can get a tree decomposition with width $2w+1$
in time $2^{O(w)} \cdot n$ due to the recent work by
Korhonen~\cite{Korhonen21}. %
For a more detailed introduction to treewidth, see e.g.\ the book by Cygan et
al.~\cite{CyganFKLMPPS15}. % 

\subsection{Basic Tools}
\label{sec:basic}

\begin{restatable}[*]{lemma}{lemBasicTransitivity}
\label{lem:basic:transitivity}
Let $G$ be a network. %
If $H$ is a quality-$q$ sparsifier for $G$ with terminals $K$, and $L$ is a quality-$r$ sparsifier for $H$ with terminals $K'$, $K' \subseteq K$, then $L$ is a quality-$qr$ sparsifier for $G$ with terminals $K'$, where the statement works if $H$, $L$ are cut sparsifiers or flow sparsifiers. %
\end{restatable}

We recall the splicing lemma of Andoni et al.~\cite{AndoniGK14} for flow
sparsifiers, which shows that it is sufficient for a sparsifier to preserve
routings along terminal-free paths.

\begin{lemma}[{\cite[Lemma 5.1]{AndoniGK14}}]
\label{lem:basic:splicing}
Let $G$ and $H$ be two networks with the same set of terminals $K$, and fix
$\rho \geq 1$. %
Suppose that whenever a demand $\bfd$ between terminals in $K$ can be routed
in $G$ using terminal-free flow paths, demand $\bfd/\rho$ can be routed in $H$
(by arbitrary flow paths).

Then for every demand $\bfd$ between terminals in $K$ that can be routed in
$G$, demand $\bfd/\rho$ can be routed in $H$.
\end{lemma}

Finally, we show the following generalization of the composition lemma to both
cut and flow sparsifiers.

\begin{restatable}[*]{lemma}{lemBasicUplus}
\label{lem:basic:uplus}
Let $G_1$ and $G_2$ be Steiner-disjoint networks with terminal set $K$.

If $H_1$ and $H_2$ are quality-$q$ (cut or flow) sparsifiers for $G_1$ and $G_2$ with terminal set $K \cap V(G_1)$, $K \cap V(G_2)$, respectively, then $H:=H_1 \uplus H_2$ is a quality-$q$ (cut or flow, resp.) sparsifier for $G:=G_1 \uplus G_2$ (parallel edges in $K$ are joined and their capacities summed).
\end{restatable}

The proof for flow sparsifiers is given by Andoni et al.~\cite[Lemma 5.2]{AndoniGK14}; the proof for cut sparsifiers follows using similar arguments.

Let $G / vw$ be the network obtained from $G$ by contracting $v$ and $w$ into
a vertex denoted $vw$, that is, removing $v$ and $w$, adding a vertex $vw$
with edges to vertices $(N(v) \cup N(w)) \setminus \set{v,w}$, and setting the
capacity of each new edge $\set{vw, x}$ to $w_{G/vw}(\set{vw,x}) = u_G({vx}) +
u_G({wx})$, where $u_G({vx}) = 0$ if the edge does not exist.

\begin{restatable}[*]{lemma}{lemBasicEquivMerge}
\label{lem:basic:equiv-merge}
Let $v$, $w$ be vertices such that for every $A \subseteq K$, there is a
min-cut $X$ separating $A$ that either contains both $v$ and $w$ or neither of
them.

Then $G / vw$ is an exact cut sparsifier for $G$ with terminals $K$.
\end{restatable}

\section{Sparsifiers for Quasi-Bipartite Graphs}
\label{sec:bipart}

In this section we show how to compute cut and flow sparsifiers for
quasi-bipartite graphs, where the left side is the set of terminals. %
In a later section we show how to improve this to a more general case of
bounded vertex cover. %
Our results are formalized in the following theorems:

\begin{theorem}
\label{thm:bipartite}
Let $G = (V,E,u)$ be a network with terminal set $K$ of size $k$.

If $G$ is bipartite with partition $V = K \uplus (V \setminus K)$, then $G$
has a cut sparsifier of size $2^{k^2}$ and a contraction-based cut sparsifier of size
$2^{k^3}$.
\end{theorem}

\begin{theorem}
\label{thm:bipartite-flow}
Let $G = (V,E,u)$ be a network with terminal set $K$ of size $k$.

If $G$ is bipartite with partition $V = K \uplus (V \setminus K)$, then $G$
has a contraction-based flow sparsifier of size $3^{k^3}$.
\end{theorem}

We remark that quasi-bipartite and bipartite graphs are equally hard to
handle, as we can simply consider a quasi-bipartite graph as the
Steiner-disjoint union of $G[K]$ and $G-E(K)$.

\begin{corollary}
Quasi-bipartite networks (with no edges between Steiner vertices) have cut sparsifiers of size $2^{k^2}$ and flow sparsifiers of size $3^{k^3}$.
\end{corollary}

\subsection{Cut Sparsifiers}
\label{sec:bipart:cut}

This section is dedicated to proving \Cref{thm:bipartite}.

Let $v$ be the center of a star, and $(c(1), c(2), \ldots, c(k))$ be the
capacity vector of the edges to the terminals (with some ordering $K=\set{t_1,
\ldots, t_k}$). %
For each cut $S \subseteq K$, either $c(S) \leq c(K-S)$, $c(S) = c(K-S)$, or
$c(S) \geq c(K-S)$. %
Using the inequalities for $c$ and each subset $S \subseteq K$, we can define
a polyhedron $P_c$ of all the capacity vectors that cut the star in the same way
as $c$. %
Thus, the capacity vector $c$ is a conic combination of the extreme rays of
$P_c$, and therefore we can replace it in the graph by a conic combination of the stars
corresponding to the extreme rays, which we call \emph{basic stars}. %
Finally, we show that since $c$ agrees on every inequality with the basic
stars in its conic combination, the replacement preserves the value of
min-cuts, and so by replacing every star we obtain a sparsifier of $G$.

\subparagraph{Basic stars.}
Let $c \in \Rbb^k$ be a capacity vector. %
We define $\Scal_c$  to be the collection of subsets $S \subseteq K$ such that $c(S) \leq c(K-S)$, i.e.\ 
$\Scal_c = \set{S \subseteq K: c(S) \leq c(K-S)}$,
and the \emph{star cone} of $c$ to be:
\[
P_c := \set[\Big]{x \in \mathbb R^k_{\geq 0}:\; x(S) \leq x(K-S)\;\forall S \in \Scal_c}
\]
We say that a vector $x$ \emph{agrees with $c$} (on every cut) if $x(S) \leq
x(K-S)$ for all $S \in \Scal_c$, and thus $P_c$ is the polyhedron containing
all of the capacity vectors that agree with $c$. %
It is a cone since it is defined by constraints of the form $\alpha^T x \leq 0$. %
We remark that for every $S \subseteq K$, $S \in \Scal_c$ or $K-S \in
\Scal_c$, and both are present if and only if $x(S) = x(K-S)$.

We define the set of basic stars as stars constructed from extreme rays of
any such cone. %
For a given $c$, the extreme rays of the cone $P_c$ are found at the
intersection of $k-1$ tight inequalities: %
 $x_i \geq 0$ for some indices $i\in I$, and ${x(S) = x(K-S)}$ for some $S \in
J$, with $|I| + |J| = k-1$ (see e.g.~\cite[Sec.~3.12]{ConfortiCZ14}). %
Notice that, regardless of the capacity vector $c$, the extreme rays of $P_c$ are all
found using a tight subset of the same collection of inequalities. %
However, not every extreme ray belongs to every $P_c$, as they might disagree
on some inequalities outside of $J$. %
For convenience, we represent each ray by a vector with coordinates summing to 1.

Let $Q$ be the set of extreme rays of any cone as obtained above, that is, the
set of extreme rays obtained from intersection of $k-1$ independent tight
constraints with $x(K)=1$. %
Formally, let $\Ical_k$ be the collection of pairs $(I,J)$, $|I|+|J| = k-1$,
such that the constraints $x(i) = 0$ for $i \in I$, $x(S) = x(K-S)$ for $S \in J$
and $x(K)=1$ are all independent. %
Then
\begin{align*}
Q = \set[\big]{q_{IJ} \in \Rbb_{\geq 0}^k:\ &q_{IJ}(K) = 1; q_{IJ}(I) = 0; \\
&q_{IJ}(S) = q_{IJ}(K-S)\; \forall S \in J; (I,J) \in \Ical_k}.
\end{align*}

For each $q \in Q$, we can construct a star with center denoted $v_q$ and an
edge to each terminal $t_i \in K$ with capacity $q_i$. % 
These stars are denoted $\emph{basic stars}$ and are referred to by their center $v_q$. %

The size of $Q$ is determined by the possible sets of inequalities that define
each of its elements. % 
As the tight inequalities for $S$ and $K \setminus S$ are the same, there are
effectively at most $2^{k-1}+k \leq 2^k$ inequalities to choose from. %
Each element of $Q$ is defined by $k-1$ of these, and thus $|Q| \leq
\binom{2^k}{k-1} \leq 2^{k^2}$.

\begin{restatable}[]{lemma}{lemBipartConicComb}
\label{lem:bipart:conic-comb}
Any capacity vector $c$ can be written as the conic combination of at most
$k$ points in $Q$, all of which agree with $c$. %

Formally, there are $q_1, \ldots, q_{k} \in Q$, $\lambda(q_1), \ldots, \lambda(q_k) \geq 0$ such that :
\begin{align*}
\sum_{i=1}^{k} \lambda(q_i) \cdot q_i &= c &&\text{and}&
q_i(S) &\leq q_i(K-S) &\forall&i \in [k], S\in \Scal_c
\end{align*}
\end{restatable}

\begin{proof}
Let $P_c$ be the star cone corresponding to $c$. %
By Carathéodory's theorem (see e.g.~\cite[Sec.~3.14]{ConfortiCZ14}), $c$ is a
conic combination of at most $k$ extreme rays of $P_c$ (which are contained in
$Q$). %
In other words, there are $q_1, \ldots, q_k \in P_c \cap Q$ and $\lambda(q_1),
\ldots, \lambda(q_k) \geq 0$ such that $c = \sum_{i=1}^{k} \lambda(q_i) \cdot
q_i$.

All it remains to show is that every $q_i$ agrees with $c$ on every cut. %
Indeed, since every $q_i$ is an extreme ray of $P_c$, it must satisfy the
inequalities defining $P_c$, and thus agree with $c$. %
\end{proof}

We obtain a sparsifier for $G$ as follows: %
for each $v \in V-K$, write the capacity vector $c_v$ of the star centered at
$v$ as a conic combination of the points in $Q$, i.e.\ as $c_v = \sum_{q \in Q}
\lambda_v(q) q$.
Then, define $V_Q = \set{v_q: q \in Q}$ and take $H = (K \cup V_Q, K \times
V_Q, u')$; the capacity vector for each $v_q$ is $q$ scaled up by the sum of the
corresponding values $\lambda_v$: $u'(v_q, \cdot) = q \cdot \sum_v \lambda_v(q)$.

To determine the values $\lambda_v$, we can use the constructive proof of
Carathéodory's theorem, starting by computing the set $Q$ by enumeration in
time $2^{k^2}\cdot\poly(k)$. %
Given a capacity vector $c$, start by finding $q_1 \in Q \cap P_c$ that
agrees with $c$ on every cut; %
then, find the maximum value of $\lambda(q_1)$ such that $c-\lambda(q_1) q_1 \in P_c$; %
finally, set $c' = c - \lambda(q_1) q_1$ and repeat the process for $c'$ to find
the remaining $q_2, \ldots$ and $\lambda_2, \ldots$ %
Notice that at each step, $c'$ has at least one more tight inequality than $c$
(otherwise we could increase $\lambda(q_1)$), and thus the process stops after
$k$ iterations. %
In conclusion, the process of computing the set $Q$ takes time $2^{O(k^2)}$, and the process
of finding the coefficients for each star takes time $2^{O(k^2)}$ (per star).

The following lemma allows us to relate a star to its conic combination using
basic stars.

\begin{restatable}[]{lemma}{lemBipartStarDecomp}
\label{lem:bipart:star-decomp}
Let $v$ be a vertex with degree $k$ and capacity vector $c$ for its incident edges. %

If  $c$ can be written as the conic combination of points $q_1, q_2, \ldots,
{q_{\ell} \in Q}$, all of which agree with $c$, %
then $G$ has a sparsifier given by $G - \set{v} + \set{w_1, \ldots,
w_{\ell}}$, where the neighbors of vertices $w_1, \ldots, w_{\ell}$ are also
neighbors of $v$ and the capacities of the edges incident on each vertex $w_i$
are given by $\lambda(q_i) \cdot q_i$.
\end{restatable}

\begin{proof}
We will use the fact that the definition of sparsifier is reflexive, and thus
show that $G$ is a sparsifier for $G - \set{v} + \set{w_1, \ldots, w_\ell}$. %
We will iteratively contract two vertices in $\set{w_1, \ldots, w_\ell}$,
until we have the original graph. %
We assume that each vertex $w_i$ has the same neighbors as $v$ by adding edges
with capacity $0$.

Let $G' = G - \set{v} + \set{w_1, \ldots, w_\ell}$ so that $u_G(v) =
\sum_{i=1}^\ell u_{G'}(w_i)$. %
For any $i \in [\ell]$, $q_i$ agrees with $c$, so $u_{G'}(w_{i}) =
\lambda(q_{i}) \cdot q_{i}$ also agrees with $c$, as $P_c$ is a cone (and thus
scale-invariant). %
We observe that if some $A \in \Scal_c$, then $c(A) \leq c(K-A)$, and thus
there is a min-cut separating $A$ that does not contain $v$, as it is at least
as cheap to cut the edges to $A$ as to $K-A$; the same argument applies to any
vertex whose capacities agree with $c$, such as any $w_i$. %
Therefore, for any $A \subseteq K$, there is a min-cut separating $A$ that
contains both $w_{\ell-1}$ and $w_\ell$ (if $K-A \in \Scal_c$) or neither (if
$A \in \Scal_c$), and thus we can apply \Cref{lem:basic:equiv-merge} to merge
$w_{\ell-1}$ and $w_{\ell}$ to obtain a new vertex $w'_{\ell-1}$ with
capacities $u_{G'}(w_{\ell-1}) + u_{G'}(w_{\ell})$. %
This ensures that $u_G(v) = \sum_{i=1}^{\ell-2} u_{G'}(w_i) +
u_{G'}(w'_{\ell-1})$, and thus we can repeat the process until we only have a
single $w'_1$ left, with capacities $u_{G'}(w'_1) = \sum_{i=1}^{\ell}
u_{G'}(w_i) = u_G(v)$, which is equivalent to the original graph.
\end{proof}

We can now finish the proof of \Cref{thm:bipartite}.

\begin{proof}[Proof of \Cref{thm:bipartite}]
We start by computing the basic stars $Q$ for $k$ terminals, and then %
computing the coefficients $\lambda_v: Q \to \Rbb_{\geq 0}$ for each $v \in V
\setminus K$. We can now obtain two different constructions for a sparsifier,
one that replaces all the vertices with basic stars, and another which simply
contracts vertices that have the same cut profile.

\subparagraph{Basic star sparsifier.} We take $V_Q = \set{v_q: q \in Q}$ as the
set of all basic stars, and construct our sparsifier by adding edges from
every $v_q \in Q$ to every terminal $t \in K$. %
The capacities are given by summing over the $\lambda_v$ as follows:
$u'(v_q, t_i) = q_i \cdot \sum_v \lambda_v(q)$. %
The sparsifier is then given by $H = (K \cup V_Q, K \times V_Q, u')$.

This construction is equivalent to decomposing each $v$ into a conic
combination of basic stars using \Cref{lem:bipart:star-decomp}, and then
(iteratively) contracting all of the vertices created for the same $q \in Q$
(and different $v \in V \setminus K$) using \Cref{lem:basic:equiv-merge}. %
This also proves correctness of the sparsifier.

\subparagraph{Contraction-based sparsifier.} For this variant of the sparsifier,
we will only use \Cref{lem:basic:equiv-merge} and Carathéodory's
theorem~\cite[Sec.~3.14]{ConfortiCZ14}, %
as well as the algorithm introduced by Hagerup et
al.~\cite{HagerupKNR98}, in which we compute all $2^k$ minimum cuts, and
contract any two vertices whose capacity vectors agree with each other. %

We can also compute the sparsifier directly in time $O(2^k \cdot n^2)$, 
by computing the set $\Scal_c$ for each
of the at most $n$ stars, and then comparing the sets for each pair of stars
to decide whether to contract them. %
Alternatively, the sets $\Scal_c$ can be computed in time $O(2^k \cdot k n)$
and the vertices placed in buckets according to their set $\Scal_c$,
after which the vertices in each bucket can be contracted.

Thus, all that we need to do is to show that if two capacity vectors are the
conic combination of the same basic stars, then they agree, and thus can be
contracted.

\begin{lemma}
\label{lem:bipart:equiv-comb}
If capacity vectors $c$ and $c'$ can be written as the conic combination of
points $q_1, q_2, \ldots, q_{\ell} \in Q$ and each $q_i$ agrees with both $c$
and $c'$, then $c$ and $c'$ agree with each other.
\end{lemma}

\begin{proof}
As the lemma is symmetric, we will simply show that if $S \in \Scal_c$, then
$S \in \Scal_{c'}$. %
If $S \in \Scal_c$, then $S \in \Scal_{q_i}$ and so $q_i(S) \leq
q_i(K\setminus S)$ by definition of agreement. %
Therefore, 
\[
c'(S) - c'(K \setminus S) 
= \sum_i \lambda_{c'}(q_i) \cdot \paren[\big]{q_i(S) - q_i(K \setminus S)}
\leq 0 
\quad \Rightarrow \quad
S \in \Scal_{c'}.
\qedhere
\]
\end{proof}

Given \Cref{lem:bipart:equiv-comb}, it is sufficient to bound the number of
possible ways that capacity vectors can be written as conic combinations, as
vectors that are written as a conic combination of the same basic stars are
always on the same minimal min-cuts, and thus can be contracted by
\Cref{lem:basic:equiv-merge}. %
As there are $2^{k^2}$ basic stars, and by Carathéodory's theorem each
capacity vector can be written as the conic combination of at most $k$ of
them, there are at most $2^{k^3}$ such combinations, and thus
after contraction, we are left with a sparsifier of size at most $2^{k^3}$.
\end{proof}

\subsection{Flow Sparsifiers}

We will show that a slight modification to the contraction-based sparsifier
in \Cref{sec:bipart:cut} increases its size to $3^{k^3}$ but makes it a flow sparsifier on bipartite graphs, proving \Cref{thm:bipartite-flow}. %

Let $c$ be a capacity vector. %
The only modification we need is to  consider all inequalities of the form
$c(A) \leq c(B)$ for $A,B \subseteq K$, $A \cap B = \emptyset$. %
Surprisingly, the cut sparsifiers of \Cref{thm:bipartite} preserve min-cuts
separating a subset $A\subseteq K$ from a disjoint subset $B\subseteq K$
without requiring these inequalities. %
However, for the construction of flow sparsifiers it is necessary that the
vertices we contract agree on the inequalities for each pair $(A,B)$.

We define $\Scal'_c = \set{(A,B) \subseteq K: c(A) \leq c(B)}$ and say
that $c'$ \emph{strongly agrees} with $c$ if $c'(A) \leq c'(B)$ for all $(A,B)
\in \Scal'_c$. %
Notice that $\Scal'_c$ has at most $3^k$ sets, as an element can be placed in
$A$, $B$ or neither. %

The \emph{strong star cone} of $c$ is defined as
$P'_c := \set[\big]{x \in \mathbb R^k_{\geq 0}:\; x(A) \leq x(B)\;\forall (A,B) \in \Scal'_c}$, and the set of extreme rays $Q' := \set{q \in \Rnn^k: \text{$q$ is an extreme ray of some $P'_c$}}$. %
The size of $Q'$ is upper-bounded by the possible combinations of $k-1$ tight inequalities, which implies that $|Q'| \leq 3^{k^2}$.

We use the contraction-based construction of \Cref{sec:bipart:cut} using the
concept of strong agreement. %
As before, we can show that each capacity vector is the conic combination of
$k$ extreme rays. %

\begin{restatable}[]{lemma}{lemBipartFlowConicComb}
\label{lem:bipart:flow:conic-comb}
Any capacity vector $c$ can be written as the conic combination of at most
$k$ points in $Q'$, all of which strongly agree with $c$. %

Formally, there are $q_1, \ldots, q_{k} \in Q'$, $\lambda(q_1), \ldots, \lambda(q_k) \geq 0$ such that :
\begin{align*}
\sum_{i=1}^{k} \lambda(q_i) \cdot q_i &= c &&\text{and}&
q_i(A) &\leq q_i(B) &\forall&i \in [k], (A,B)\in \Scal'_c
\end{align*}
\end{restatable}

\begin{proof}
Let $P'_c$ be the star cone corresponding to $c$. %
By Carathéodory's theorem, $c$ is a conic combination of at most $k$ extreme
rays of $P'_c$ (contained in $Q'$). %
Each $q_i$ is in $P'_c$, and thus strongly agrees with $c$.
\end{proof}

As $Q'$ has at most $3^{k^2}$ extreme rays, this means that there are at most 
$3^{k^3}$ classes where capacity vectors can be placed according to which set
of extreme rays produce it as a conic combination. %
As part of our proof, we need to show that if two capacity vectors are in the
same class, then they strongly agree.

\begin{restatable}[]{lemma}{lemBipartFlowEquivComb}
\label{lem:bipart:flow:equiv-comb}
If capacity vectors $c$ and $c'$ can be written as the conic combination of
points $q_1, q_2, \ldots, q_{\ell} \in Q'$ and both strongly agree with each
$q_i$ on every cut, then $c$ and $c'$ strongly agree.
\end{restatable}

\begin{proof}
As the lemma is symmetric, we will simply show that if $(A,B) \in \Scal'_c$, then
$(A,B) \in \Scal_{c'}$. %
If $(A,B) \in \Scal_c$, then $(A,B) \in \Scal_{q_i}$ and so $q_i(A) \leq
q_i(B)$ by definition of agreement. %
Therefore, 
\[
c'(A) - c'(B) 
= \sum_i \lambda_{c'}(q_i) \cdot \paren[\big]{q_i(A) - q_i(BS)}
\leq 0 
\quad \Rightarrow \quad
(A,B) \in \Scal_{c'}.
\qedhere
\]
\end{proof}

We then show that if we have two stars where their capacity vectors agree, we
can contract them to obtain a flow sparsifier. %
By repeating the process until we have at most 1 vertex per class, we obtain a
flow sparsifier of size $3^{k^3}$.

\begin{lemma}
\label{lem:bipart:flow:equiv-merge}
Let $v_1$, $v_2$ be centers of stars with leaves $K$ and let $c_1, c_2 \in
\Rnn^k$ be their capacity vectors (respectively). %

If $c_1$ and $c_2$ strongly agree with each other, that is $\Scal'_{c_1} =
\Scal'_{c_2}$, then $G / v_1v_2$ is an exact flow sparsifier for $G$ with
terminals $K$.
\end{lemma}

\begin{proof}
Let $c := c_1 + c_2$, and $v := v_1v_2$ be the vertex created in $G':=G / v_1v_2$. %
We will show that if we can route a demand $\bfd$ in $G$ then we can route it
in $G'$ and vice-versa. %
If we can route a demand in $G$, then any demand routed on $v_1w$ or $v_2w$
can be routed instead on $vw$, as $c_1(w) + c_2(w) = c(w)$. %
We now show that a demand routed in $G'$ can also be routed in $G$. %

By \Cref{lem:basic:splicing}, and since the neighbors of $v_1$ and $v_2$ are
terminals, we can ``splice'' the flows so that any demand is routed only
through (internally) terminal-free paths. Thus, we focus on the demands that
are routed through paths $t_i v t_j$, $t_i, t_j \in K$ in $G'$. %
Specifically, if $\bfd$ is a demand that can be routed in $K \cup
\set{v}$, then we construct demand vectors $\bfd_1$, $\bfd_2$ that
can be routed through $v_1$, $v_2$, such that $\bfd_1 + \bfd_2 = \bfd$. %

\begin{observation}
In a star $K \cup \set{v}$ with capacities $c$, a demand $\bfd$ can be routed
if and only if for every $i \in [k]$, $\bfd(i):= \sum_j d(t_i,t_j) \leq c(i)$.
\end{observation}

\begin{figure}
  \begin{subfigure}[t]{0.5\textwidth}
    \centering
    \setlength{\unitlength}{157bp}%
    \begin{picture}(1,0.928)%
      \put(0,0){\includegraphics[width=\unitlength]{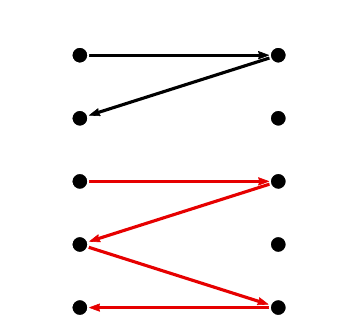}}%
      \put(0.19,0.83){$K_1$}%
      \put(0.74,0.83){$K_2$}%
      \put(0.42,0.81){$d_1(i,j)$}%
      \put(0.42,0.61){$d_2(j,i')$}%
      \put(0.07,0.72){$c_1(i)$}%
      \put(0.80,0.72){$c_2(j)$}%
      \put(0.47,0.44){\color[rgb]{0.90196078,0,0}$P$}%
    \end{picture}%
  \end{subfigure}\hfill
  \begin{subfigure}[t]{0.5\textwidth}
    \centering
    \setlength{\unitlength}{157bp}%
    \begin{picture}(1,0.928)%
      \put(0,0){\includegraphics[width=\unitlength]{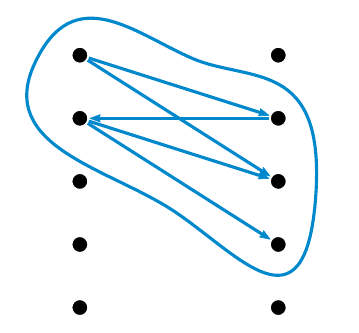}}%
      \put(0.46,0.83){\color[rgb]{0,0.53,0.8}$d'(X) = c_B(X)$}%
    \end{picture}
  \end{subfigure}
  \caption{Representation of the graph $G_B$; a path (left, red) as found in
  Step \ref{lem:bipart:flow:equiv-merge:3}; a cut (right, blue) representing
  the situation in which no path is found}\label{lem:bipart:flow:demand-graph}
\end{figure}
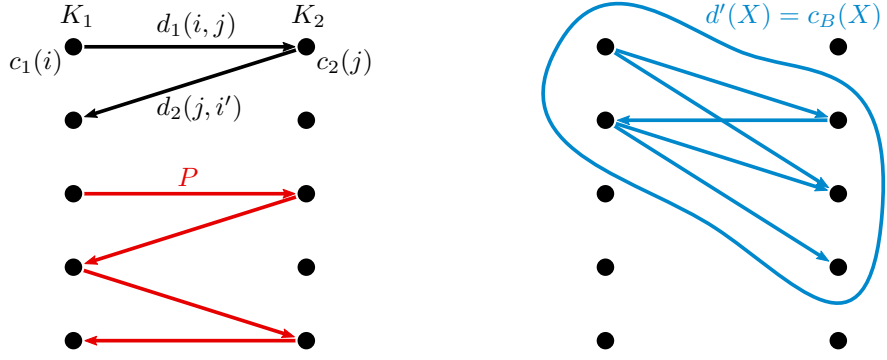

We use the following directed bipartite graph $G_B$ to assist us in the
algorithm and proof: $G_B$ has as vertices two copies of $K$, $V(G_B) = K_1
\cup K_2$, and has arcs $A(G_B) = (K_1 \times K_2) \cup (K_2 \times K_1)$. %
We also impose node capacities for outgoing arcs, with capacity vector $c_1$
for the vertices in $K_1$ and $c_2$ for vertices in $K_2$; we represent the
capacity vector for $G_B$ as $c_B$ (so $c_B(K_1) = c_1$, $c_B(K_2) = c_2$). %
See \Cref{lem:bipart:flow:demand-graph} for a visual representation of the graph.

Our goal is to add demand to the graph in the form of demands on arcs, with
arcs from $K_1$ to $K_2$ representing $\bfd_1$, and arcs from $K_2$ to $K_1$
representing $\bfd_2$. %
We will use $\bfd'$ to represent the demand in the graph (so $\bfd_1 =
\bfd'(K_1,K_2)$, $\bfd_2 = \bfd'(K_2,K_1)$), and $\bfdt$ to represent
leftover demand (initially $\bfd' = 0$, $\bfdt = \bfd$). %
Whenever we say to add some demand to $d_\iota(i,j)$, it is implied that we
add the demand to $d'(i_\iota,j_{3-\iota})$ and remove it from $\dt(i,j)$, as
well as make the same changes to $(j,i)$ (increase $d_\iota(j,i)$ and
$d'(j_\iota, i_{3-\iota})$, decrease $\dt(j,i)$), where $i_1, j_1 \in K_1$,
$i_2, j_2 \in K_2$ are the copies of $i$ and $j$ in $K_1$, $K_2$,
respectively.

We repeat the following steps until $\bfdt = 0$:
\begin{enumerate}
\item if there is $\dt(i,j) > 0$ such that $\bfd_1(i) < c_1(i)$ and $\bfd_1(j)
< c_1(j)$, add to $d_1(i,j)$ a value of $\epsilon := \min(\dt(i,j), c_1(i)-\bfd_1(i), c_1(j)-\bfd_1(j))$;
\item similarly for $\bfd_2$, if there is $(i,j)$ such that $\epsilon :=
 \min(\dt(i,j), c_2(i)-\bfd_2(i), c_2(j)-\bfd_2(j)) > 0$, add $\epsilon$ to
 $d_2(i,j)$; 
\item let $\dt(i,j) > 0$, and assume w.l.o.g.\ that $\bfd_2(i) = c_2(i)$: %
    \begin{enumerate}
    \item  find a path ${P=\paren{i_2 = \ell_0, \ell_1, \ell_2, \ldots
        \ell_p}}$ such that for all $0<r<p$, ${\bfd'(\ell_r) = c_B(\ell_r)}$,
        ${d'(\ell_r, \ell_{r+1}) > 0}$, and for the endpoint, $d'(\ell_p) <
        c_B(\ell_p)$; %
        \label{lem:bipart:flow:equiv-merge:3}
    \item switch demand between $\bfd_1$ and $\bfd_2$ as follows: let %
        $\epsilon := \min\paren[\big]{c_1(i)-\bfd_1(i),{c_B(\ell_p)-\bfd'(\ell_p)},\linebreak \min_r
        d'(\ell_r, \ell_{r+1})};$ %
        for any arc $(x_1, y_2) \in K_1 \times K_2$ (resp. $(x_2, y_1)\in K_2
        \times K_1$) in $P$, decrease $d'(x_1, y_2)$ (resp. $d'(x_2, y_1)$) and
        increase $d'(x_2, y_1)$ (resp. $d'(y_1, x_2)$) by $\epsilon$; %
    \item add $\min(\dt(i,j), \epsilon, c_2(j)-\bfd_2(j))$ to $d_2(i,j)$.
    \end{enumerate}
\end{enumerate}

We will show that such a path always exists as long as there is demand to be
routed, but first, let us show that these operations maintain the invariants
that $\bfd'(i) \leq c_B(i)$, for all $i \in V(G_B)$. %
For the first two steps, $\bfd_1(i)$ (resp.\ $\bfd_2(i)$) increases by at most
$c_1(i)-\bfd_1(i)$ (resp.\ $c_2(i)-\bfd_2(i)$), which maintains the invariant. %
For the third step, notice that every $\ell_i$ except the endpoints has two
arcs in the path, one corresponding to $\bfd_1$ and the other to $\bfd_2$,
hence the decrease on $\bfd_1$ on one edge is compensated by the increase on
the other, and the same for $\bfd_2$. %
Furthermore, as we choose $\epsilon$ to be the smallest value on the path, no
demand can go below $0$. %
Finally, for the endpoints, we know that $\ell_p$ has spare capacity by
definition, and for $i$ we decrease $\bfd_2(i)$ and increase $\bfd_1(i)$, but
$\bfd_2(i) = c_2(i)$ implies that $\bfd_1(i) < c_1(i)$, as $\bfdt(i) +
\bfd_1(i) + \bfd_2(i) = \bfd(i) \leq c(i) = c_1(i) + c_2(i)$ and $\bfdt(i) >
0$.

We will now show by contradiction that a path must exist. %
Assume that the process above cannot complete, i.e.\ there is some demand
$\dt(i,j)>0$ that cannot be placed in $\bfd_1$ or $\bfd_2$, and there is no
path $P$ as specified above. %
Then there is a set $X \subseteq V(G_B)$ of vertices reachable from $i$ by
edges with positive demand in $\bfd'$, all of which have saturated capacity
$\bfd'(\ell_r) = c_B(\ell_r)$. %
Writing $X_1 = X \cap K_1$, $X_2 = X \cap K_2$, we get that $\bfd_1(X_1) =
\sum_{\ell \in X_1} \bfd_1(\ell) = c_1(X_1)$ and $\bfd_2(X_2) = c_2(X_2)$. %
We can furthermore deduce that $\bfd_1(X_1) \leq \bfd_1(X_2)$, since all of
the demand in $\bfd_1$ incident on $X_1$ is represented as demand in an arc
$(X_1,X_2)$, which is thus also incident on $X_2$. %
Similarly, we know that $\bfd_2(X_2) \leq \bfd_2(X_1)$. %
Putting these facts together, we conclude that:
\begin{alignat*}{5}
c_1(X_1) &= \bfd_1(X_1)\, &\leq \bfd_1(X_2) &\leq c_1(X_2) &&\quad\text{and}\quad&
c_2(X_2) &= \bfd_2(X_2)\, &\leq \bfd_2(X_1) &\leq c_2(X_1)
\end{alignat*}
Since $\Scal'_{c_1} = \Scal'_{c_2}$, it must be the case that $c_1(X_1) =
c_1(X_2)$ and $c_2(X_2) = c_2(X_1)$, and thus from the inequalities above we
can conclude that $\bfd_1(X_2) = c_1(X_2)$, and in particular, since $i_2 \in
X_2$, it must be that $\bfd_1(i) = c_1(i)$. %
But this is a contradiction, as $\bfd(i) = \bfdt(i) + \bfd_1(i) + \bfd_2(i) =
\bfdt(i) + c_1(i) + c_2(i) > c(i)$, and thus $\bfd$ would not be routable in $G'$.

We conclude that, as long as $\bfdt \neq \mathbf{0}$, the process above finds
a path and thus makes progress in each iteration. %
Therefore, $\bfd$ is routable in $G$ by splitting it into demands $\bfd_1$ for
$v_1$ and $\bfd_2$ for $v_2$ as described above.
\end{proof}

\section{Sparsifiers for Small Vertex Cover and Integrity}
\label{sec:vi}

The vertex cover number $c$ of a graph $G$ is the size of the smallest set $X$
such that $G-X$ contains no edges. %
The vertex integrity $d$ of a graph extends this by allowing $G-X$ to have
small components, and it is the smallest number for which there is a set $X$
such that $G-X$ has components of size at most $d-|X|$ (i.e.~the size of a
component plus $X$ does not exceed $d$). %
We show that if either of these parameters is bounded, the exponential complexity in the size of the
sparsifier is limited to an \emph{additive} term depending only on
the parameter. %
The result of \Cref{sec:bipart:cut} corresponds to the case of $c =
k$ (or $d = k+1$).

Our formal results are the following:

\begin{restatable}{theorem}{thmVc}
\label{thm:vc}
Let $G = (V,E,u)$ be a network with terminal set $K$ of size $k$.

If $G$ has a vertex cover of size $c$, it has a cut sparsifier of size $k +
2^{c^2}$ and a flow sparsifier of size $k+3^{c^3}$.
\end{restatable}

\begin{proof}
Let $X$ be a vertex cover of size $c$. %
We will start by splitting  the graph into a subgraph with all of the
terminals and another containing only $X$ and non-terminals. %
Let $G_K = G[K \cup X]$ be the first of these graphs, and $G_S = {G[(V\setminus K) \cup X] -
E(X)]}$ be the second. %
Notice that the graphs are Steiner-disjoint for terminal set $K \cup X$, and
thus we can use \Cref{lem:basic:uplus}. %
Furthermore, if we compute a sparsifier $H_S$ for $G_S$ (with terminal set
$X$) and take the sparsifier $H = G_K \uplus H_S$ for $G$, the size of $H$ is
$k + |V(H_S)|$. %
All that remains is to apply \Cref{thm:bipartite} to obtain a cut sparsifier
$H_S$ of size $2^{c^2}$ or \Cref{thm:bipartite-flow} to obtain a flow
sparsifier $H_S$ of size $3^{c^3}$, which completes the proof of the theorem.
\end{proof}

\begin{theorem}
\label{thm:vi-v2}
Let $G = (V,E,u)$ be a network with terminal set $K$ of size $k$.

If $G$ has a separator $X\subseteq V$ of size $a$ and $|C| \leq b$ for every
component $C$ of $G-X$, then it has a sparsifier of size $kb + 4^{b(a+b)^2}$.
\end{theorem}

As a corollary, we get that  small sparsifiers exist when vertex integrity is
bounded.

\begin{corollary}
\label{thm:vi}
Let $G = (V,E,u)$ be a network with terminal set $K$ of size $k$.

If $G$ has a separator $X\subseteq V$ such that $|X| + |C| \leq d$ for every
component $C$ of $G-X$, then it has a cut sparsifier of size $kd + 4^{d^3}$.
\end{corollary}

The rest of the section is dedicated to proving \Cref{thm:vi-v2}.

We use the same strategy of separating out the terminals, though in
this case we need to show a new bound when $G-X$ has small connected
components. %
We will use similar polyhedral techniques adjusted for the case of small
components.

Let $X \subseteq V$ be a separator of size $a$ such that $G-X$ has
components of size at most $b$. %
Let $\Ccal = \operatorname{cc}(G-X)$, and let $\Ccal_K$ be the set of connected components
containing any terminals. %
We define $G_K = G\brac{X \cup \bigcup \set{C \in \Ccal_K}}$ and $G_S =
G\brac{(V(G) \setminus V(G_K)) \cup X}$. %
The graphs are Steiner-disjoint for terminals $X \cup K$, and $|V(G_K)
\setminus V(G_S)| \leq k b$, thus by \Cref{lem:basic:uplus}, $H = G_K \uplus
H_S$ is a sparsifier for $G$ of size $kb + |V(H_S)|$, where $H_S$ is a
sparsifier for $G_S$. %

All that is left is then to show that $G_S$ with terminal set $X$ has a
sparsifier of size $4^{b(a+b)^2}$. %
Though on stars (components of size 1), there are only two possible cuts for
each subset of terminals, here there are more possibilities for cuts, and thus
we need multiple inequalities for each subset of terminals. %
Therefore, we show that there are at most $4^{b(a+b)^2}$ \emph{basic
components} of size at most $b$, such that any connected component can be
written as the conic combination of these. %
We then replace each component by a conic combination of basic components, and
contract all of the components of the new graph corresponding to the same
basic component.

\subparagraph{Basic components.}
Let $C$ be a component with vertices $v_1, \ldots, v_b$, and let the separator vertices
be ordered $X=\set{t_1, \ldots, t_a}$. %
To simplify the analysis, we consider that every component has size $b$ by
adding isolated vertices if needed. %
We represent the edge capacities for $C$ as a vector $c \in \Rbb_{\geq 0}^{C
\times (C \cup X)}$ (where $c(v_i, v_j) = c(v_j, v_i)$ for all $v_i,v_j \in C$, $c(v_i, v_i) = 0$ for all $v_i \in C$). %

We define $\Scal_c$  to be the collection of triples $(A,B,B')$, $A \subseteq
X$, $B, B' \subseteq V(C)$, such that $c(A \cup B) \leq c(A \cup B')$, i.e.\
$\Scal_c = \set{(A,B,B') \subseteq X \times V(C)^2: c(A \cup B) \leq c(A \cup
B')}$, and the \emph{cut cone} of $c$ to be:
\begin{align*}
P_c := \set[\Big]{x \in \mathbb R_{\geq 0}^{C \times (C \cup X)}:\; 
    &x(A \cup B) \leq x(A \cup B')\;\forall (A,B,B') \in \Scal_c;\; \\
    &x(v_i, v_j) = x(v_j, v_i), x(v_i, v_i) = 0\;\forall v_i,v_j \in V(C)
    }
\end{align*}
We say that a vector $x$ \emph{agrees with $c$} if $c(A \cup B) \leq c(A \cup
B')$ for all $(A,B,B') \in \Scal_c$, and thus $P_c$ is the polyhedron
containing all of the capacity vectors that agree with $c$. %
It is a cone since it is defined by constraints of the form $\alpha^T x \leq 0$. %

We can now define the set $Q$ of extreme rays similarly to the star case, as
the set of all capacity vectors that satisfy $b(a+(b-1)/2)$ many linearly
independent constraints. %
For each $q \in Q$, we can construct a corresponding \emph{basic component}
$C_q$ with vertices ${X \cup \set{v_{q,1}, v_{q,2}, \ldots, v_{q,b}}}$ and
capacities defined according to $q$.

The size of $Q$ is determined by the possible sets of inequalities that define
each of its elements. % 
Out of the $2^{a+2b}$, we choose at most $b(a+b/2)$, which gives us an upper
bound of $|Q| \leq 2^{b(a+2b)(a+b)} \leq 4^{b(a+b)^2}$.

The following lemma expresses components as conic combinations of
basic components:

\begin{restatable}[]{lemma}{lemViConicComb}
\label{lem:vi:conic-comb}
Any capacity vector $c$ for a component $C$ can be written as the conic
combination of at most $\ell=b(a+b/2)$ points in $Q$, all of which agree with $c$. %

Formally, there are $q_1, \ldots, q_{\ell} \in Q$, $\lambda(q_1), \ldots, \lambda(q_\ell) \geq 0$ such that :
\begin{align*}
\sum_{i=1}^{\ell} \lambda(q_i) \cdot q_i &= c &&\text{and}&
q_i(A \cup B) &\leq q_i(A \cup B') &\forall&i \in [\ell], (A,B,B')\in \Scal_c
\end{align*}
\end{restatable}

\begin{proof}
The proof mimics that of \Cref{lem:bipart:conic-comb}, with some slight
adjustments to consider the different definition of $\Scal_c$, $P_c$ and $Q$.

Let $P_c$ be the cut cone corresponding to $c$. %
Notice that $P_c$ always includes equality constraints to ensure symmetry
($x(v_i, v_j) = x(v_j, v_i)$, $x(v_i, v_i) = 0$), and these $\binom{b}{2}+b$
constraints are independent. %
Thus, by Carathéodory's theorem $c$ is a conic combination of at most $\ell=ab
- \binom{b}{2} -b \leq b(a+b/2)$ extreme rays of $P_c$ (which are contained in
$Q$). %
In other words, there are $q_1, \ldots, q_\ell \in P_c \cap Q$ and $\lambda(q_1),
\ldots, \lambda(q_\ell) \geq 0$ such that $c = \sum_{i=1}^{\ell} \lambda(q_i) \cdot
q_i$.

Furthermore, each $q_i$ agrees with $c$ on every cut since it is an extreme
ray of $P_c$, and thus satisfies its inequalities.
\end{proof}

We obtain a sparsifier $H$ for $G$ as follows: %
given a separator $X$, $|X| \leq a$ such that $G-X$ only has connected
components of size at most $b$, we start by partitioning $G$ into $G_K =
G\brac{X \cup \bigcup \set{C \in \Ccal_K}}$ and $G_S = G\brac{(V(G) \setminus
V(G_K)) \cup X}$. %
We then compute $Q$  and the coefficients $\lambda_C$ for each component $C$
of $G_S - X$, in time $2^{O(b(a+b)^2)}$. %
Then, we define the collection of graphs $\Ccal_Q(G_S) = \set{C'_q: q \in Q}$,
where $C'_q = C_q \cdot \sum_{C \in \Ccal} \lambda_C(q)$ is $C_q$ with the
capacities scaled up by $\sum_{C \in \Ccal} \lambda_C(q)$. %
The sparsifier is obtained by doing the Steiner-disjoint union of the graphs
in $\Ccal_Q(G_S)$ and $G_K$, that is, $H = G_K \uplus \biguplus_{C'_q \in
\Ccal_Q(G_S)} C'_q$.

The analysis follows similarly to \Cref{sec:bipart:cut}, by applying
\Cref{lem:vi:conic-comb,lem:vi:sparsifier} to $G_S$.

\begin{restatable}[*]{lemma}{lemViSparsifier}
\label{lem:vi:sparsifier}
Let $G$ be a network with a separator $X\subseteq V(G)$ of size $a$ and $|C| \leq b$ for every
component $C$ of $G-X$.

Then $H = G_K \uplus \biguplus_{C'_q \in \Ccal_Q(G_S)} C'_q$ is a sparsifier
for $G$ of size at most $4^{b(a+b)^2}$.
\end{restatable}

\section{Reduction Theorem for Bounded-Treewidth Graphs}
\label{sec:btw}

In this section, we show that given a graph $G$ with $k$ terminals and a tree
decomposition for $G$ of width $w$, we can compute in linear time a sparsifier
for $G$ with size linear in $k$. %
In particular, by using as a black-box algorithm for computing a sparsifier
with size $S(k)$ and quality $g(k)$, we obtain a sparsifier of size $O(k)
\cdot S(2w)$ and quality $g(2w)$, which can be computed with $O(k)$ calls to
the original sparsifier algorithm. %
The result is formalized in \Cref{thm:reduction-tw}.

\begin{theorem}
\label{thm:reduction-tw}
Let $G = (V,E,u)$ be a network with terminal set $K$ of size $k$.
Let $S: \Nbb \to \Nbb$, $g \colon \Nbb \to \Rbb_{\geq 1}$ be functions such that every network (of treewidth $w$) has a quality-$g(k)$ cut (resp.\ flow) sparsifier of size $S(k)$.

Then any network of treewidth at most $w$ has a cut (resp.\ flow) sparsifier with quality $g(2w)$ and size $O(k) \cdot S(2w)$. %

Furthermore, given a tree decomposition of width $w$, the sparsifier can be
computed in time $O(nw)$ plus $O(k)$ calls to an algorithm computing
sparsifiers on edge-disjoint subgraphs of $G$ with at most $2w$ terminals.
\end{theorem}

As an immediate consequence of \Cref{thm:reduction-tw} and work on flow
sparsifiers~\cite{CharikarLLM10,EnglertGKRTT14,MakarychevM10}, we obtain the
following corollary:

\begin{corollary}
Any network with treewidth at most $w$ has a quality-$O(\log w)$ flow sparsifier of size $O(k \cdot w)$, %
which can be computed in polynomial time.
\end{corollary}

Let $(T, \Bcal)$ be a tree decomposition for $G$ of width $w$. %
We recall that we associate each edge with a single bag, and that we assume
that no two identical bags exist.

Let $Y\subseteq V(T)$ be a subset of bags obtained as follows: first, for each
terminal $t \in K$, add to $Y$ a node $u \in V(T)$ containing $t$, i.e.~$t \in
B_u$; then add the lowest common ancestors of any two nodes $Y$ to $Y$ as well. %

The algorithm first constructs the set $Y$, and then partitions $T$ into a set
of regions $\Rcal(T,Y)$ as follows: consider the components of $T-Y$, and
group them into regions according to their neighboring nodes. %
Finally, it returns the sparsifier $H = G[Y] \uplus \biguplus_{R \in
\Rcal(T,Y)} H_R$, where $H_R$ is a sparsifier for the subgraph $G[R]$ with
terminal set $G[R] \cap B(Y)$, computed using the black-box algorithm.

Since $T$ has $O(n)$ nodes, constructing $Y$ as well as the graphs $G[R]$ for
every region can be done in time $O(nw)$, and computing the sparsifiers simply
requires $|\Rcal(T,Y)|$ calls to the given sparsifier algorithm. %
We remark that the calls to the sparsifier algorithm are run on subgraphs
$G[R]$ for disjoint $R$, and thus induce edge-disjoint subgraphs of $G$.

We will now show that the size of $Y$ and $\Rcal(T,Y)$ is bounded, before
showing that each $G[R]$ has a small sparsifier, and that all of these can be
joined into $H$.

\begin{restatable}[]{lemma}{lemBtwRegions}
\label{lem:btw:regions}
There are at most $2k$ nodes in $Y$ and $2|Y|$ regions in $\Rcal(T,Y)$, with
each region neighboring at most two nodes of $Y$.
\end{restatable}

\begin{proof}% 
Let $T'$ be the tree obtained from $T$ by iteratively contracting every edge
that does not connect two nodes in $Y$. %
$T'$ has (at most) one node for every terminal in $K$, plus nodes for
the lowest common ancestors. %
In particular, the nodes for the lowest common ancestors have at least
2 children, as they were added to $Y$ because there are terminals in two of its
children subtrees. %
As there are at most $k$ nodes with at most 1 child, $T'$
must have at most $2k-1$ nodes.

We now show that every region neighbors at most two nodes of $Y$. %
Assume that there is a region $R$ neighboring more than 2 nodes in $Y$. %
If there are two nodes $u,v$ such that one is not the ancestor of the other,
then their lowest common ancestor $a$ is also in $Y$ and in $R$, since the
$u$-$v$-path in $T$ is contained in $R$. %
But then $a$ is a cut in $T$ which splits $R$ into smaller regions, and thus
$R$ cannot be a region in $\Rcal(T,Y)$. %
The remaining possibility is that the nodes neighboring $R$ all are ancestors
or descendants of one another, and thus are all contained in a root-leaf path
in $T$. %
However, in that case, the middle points (not the highest or lowest neighbor
of $R$) are all vertex cuts which would split $R$ into smaller regions. We
conclude that each region neighbors at most two nodes of $Y$, and that one is
the parent of the other in $T'$.

As a consequence of the above, there can be at most $|Y|$ regions that have only
a single node as neighbor, and $|Y|$ regions that have two neighbors in $Y$ (a
node $u$ and its parent).
\end{proof}

We now show that we can apply a sparsifier to $G[R]$ to obtain a good sparsifier.

\begin{restatable}[]{lemma}{lemBtwSparsifier}
\label{lem:btw:sparsifier}
Given $R \in \Rcal(T,Y)$, its induced subgraph $G[R]$ with terminal set $G[R]
\cap B(Y)$ has a sparsifier $H_R$ of quality $g(2w)$ and size $S(2w)$.
\end{restatable}

\begin{proof}% 
As $G$ has treewidth $w$, any of its subgraphs has treewidth $w$ as well. %
Thus, we only need to prove that $|G[R] \cap B(Y)| \leq 2w$, as we assume that
every graph of treewidth $w$ has a quality-$g(k)$ sparsifier of size $S(k)$,
with $k = |G[R] \cap B(Y)|$ for $G[R]$. %

If $R$ has only one neighbor $y \in V(T)$, then $|G[R] \cap B(Y)| = |G[R] \cap
B_y| \leq w+1 \leq 2w$, as the only neighboring bag to $R$ is $y$ and thus
$G[R]$ and $G[T-R]$ only intersect in $B_y$.

Otherwise, $R$ has two neighbors $y_1, y_2 \in V(T)$. %
Let $u_1$, $u_2$ be the nodes of $R$ neighboring $y_1$, $y_2$, respectively.
Notice that $R$ must correspond to a single connected component of
$T-Y$, as there is a single path connecting $y_1$ and $y_2$ in $T$,
and thus $y_1$ and $y_2$ each have a single neighbor in $R$. %

By the properties of the tree decomposition, any vertex $v$ that is
simultaneously in $G[R]$ and $G[T-R]$ must either be contained in $B_{y_1}$
and $B_{u_1}$, or be contained in $B_{y_2}$ and $B_{u_2}$, %
as the subtree induced by the bags containing $v$ must connect $R$ and $T-R$
and thus must contain the edge $y_1u_1$ or the edge $y_2u_2$. %
Since no two bags are the same, $\card{B_{y_1} \cap B_{u_1}} \leq w$,
$\card{B_{y_2} \cap B_{u_2}} \leq w$, and thus $|G[R] \cap B(Y)| \leq
\card{B_{y_1} \cap B_{u_1}} + \card{B_{y_2} \cap B_{u_2}} \leq 2w$.
\end{proof}

We complete the proof by taking the sparsifier $H = G[Y] \uplus \biguplus_{R \in
\Rcal(T,Y)} H_R$. %
$H$ is a sparsifier for $G$ because $G = G[Y] \uplus \biguplus_{R \in
\Rcal(T,Y)} G[R]$ and, by \Cref{lem:btw:sparsifier}, each $H_R$ is a
sparsifier for $G[R]$, thus by (repeated application of)
\Cref{lem:basic:uplus}, $H$ is a sparsifier for $G$ with terminal set $B(Y)
\supseteq K$. %
The size of $H$ is at most $2k(w+1) + 2kS(2w) = O(k) \cdot S(2w)$, by using
\Cref{lem:btw:regions} and it can be computed in running time
$O(nw)$ plus $O(k)$ calls to edge-disjoint subgraphs of $G$. %

\bibliography{main}

\begin{thebibliography}{10}

\bibitem{AbrahamDKKP16}
Ittai Abraham, David Durfee, Ioannis Koutis, Sebastian Krinninger, and Richard
  Peng.
\newblock On fully dynamic graph sparsifiers.
\newblock In Irit Dinur, editor, {\em {IEEE} 57th Annual Symposium on
  Foundations of Computer Science, {FOCS} 2016, 9-11 October 2016, Hyatt
  Regency, New Brunswick, New Jersey, {USA}}, pages 335--344. {IEEE} Computer
  Society, 2016.
\newblock \href {https://doi.org/10.1109/FOCS.2016.44}
  {\path{doi:10.1109/FOCS.2016.44}}.

\bibitem{AndoniGK14}
Alexandr Andoni, Anupam Gupta, and Robert Krauthgamer.
\newblock Towards {(1} + \emph{{\ensuremath{\epsilon}}})-approximate flow
  sparsifiers.
\newblock In Chandra Chekuri, editor, {\em Proceedings of the Twenty-Fifth
  Annual {ACM-SIAM} Symposium on Discrete Algorithms, {SODA} 2014, Portland,
  Oregon, USA, January 5-7, 2014}, pages 279--293. {SIAM}, 2014.
\newblock \href {https://doi.org/10.1137/1.9781611973402.20}
  {\path{doi:10.1137/1.9781611973402.20}}.

\bibitem{ArnborgCP87}
Stefan Arnborg, Derek~G. Corneil, and Andrzej Proskurowski.
\newblock Complexity of finding embeddings in a k-tree.
\newblock {\em SIAM Journal on Algebraic Discrete Methods}, 8(2):277--284,
  1987.
\newblock \href {https://doi.org/10.1137/0608024} {\path{doi:10.1137/0608024}}.

\bibitem{BarefootES87}
C.~A. Barefoot, R.~C. Entringer, and H.~C. Swart.
\newblock Vulnerability in graphs — a comparative survey.
\newblock {\em Journal of Combinatorial Mathematics and Combinatorial
  Computing}, 1, 1987.

\bibitem{BenczurK00}
Andr{\'{a}}s~A. Bencz{\'{u}}r and David~R. Karger.
\newblock Augmenting undirected edge connectivity in {Õ}(n\({}^{\mbox{2}}\))
  time.
\newblock {\em J. Algorithms}, 37(1):2--36, 2000.
\newblock \href {https://doi.org/10.1006/JAGM.2000.1093}
  {\path{doi:10.1006/JAGM.2000.1093}}.

\bibitem{Bodlaender96}
Hans~L. Bodlaender.
\newblock A linear-time algorithm for finding tree-decompositions of small
  treewidth.
\newblock {\em {SIAM} J. Comput.}, 25(6):1305--1317, 1996.
\newblock \href {https://doi.org/10.1137/S0097539793251219}
  {\path{doi:10.1137/S0097539793251219}}.

\bibitem{BodlaenderGP23}
Hans~L. Bodlaender, Carla Groenland, and Michal Pilipczuk.
\newblock Parameterized complexity of binary {CSP:} vertex cover, treedepth,
  and related parameters.
\newblock In Kousha Etessami, Uriel Feige, and Gabriele Puppis, editors, {\em
  50th International Colloquium on Automata, Languages, and Programming,
  {ICALP} 2023, July 10-14, 2023, Paderborn, Germany}, volume 261 of {\em
  LIPIcs}, pages 27:1--27:20. Schloss Dagstuhl - Leibniz-Zentrum f{\"{u}}r
  Informatik, 2023.
\newblock \href {https://doi.org/10.4230/LIPICS.ICALP.2023.27}
  {\path{doi:10.4230/LIPICS.ICALP.2023.27}}.

\bibitem{CalinescuKR04}
Gruia C{\u{a}}linescu, Howard~J. Karloff, and Yuval Rabani.
\newblock Approximation algorithms for the 0-extension problem.
\newblock {\em {SIAM} J. Comput.}, 34(2):358--372, 2004.
\newblock \href {https://doi.org/10.1137/S0097539701395978}
  {\path{doi:10.1137/S0097539701395978}}.

\bibitem{CharikarLLM10}
Moses Charikar, Tom Leighton, Shi Li, and Ankur Moitra.
\newblock Vertex sparsifiers and abstract rounding algorithms.
\newblock In {\em 51th Annual {IEEE} Symposium on Foundations of Computer
  Science, {FOCS} 2010, October 23-26, 2010, Las Vegas, Nevada, {USA}}, pages
  265--274. {IEEE} Computer Society, 2010.
\newblock \href {https://doi.org/10.1109/FOCS.2010.32}
  {\path{doi:10.1109/FOCS.2010.32}}.

\bibitem{ChaudhuriSWZ00}
Shiva Chaudhuri, K.~V. Subrahmanyam, Frank Wagner, and Christos~D. Zaroliagis.
\newblock Computing mimicking networks.
\newblock {\em Algorithmica}, 26(1):31--49, 2000.
\newblock \href {https://doi.org/10.1007/S004539910003}
  {\path{doi:10.1007/S004539910003}}.

\bibitem{abs-2407-10852}
Yu~Chen and Zihan Tan.
\newblock Cut-preserving vertex sparsifiers for planar and quasi-bipartite
  graphs.
\newblock {\em CoRR}, abs/2407.10852, 2024.
\newblock \href {https://arxiv.org/abs/2407.10852} {\path{arXiv:2407.10852}},
  \href {https://doi.org/10.48550/ARXIV.2407.10852}
  {\path{doi:10.48550/ARXIV.2407.10852}}.

\bibitem{ChenT24a/soda}
Yu~Chen and Zihan Tan.
\newblock On {(1} + {\(\varepsilon\)})-approximate flow sparsifiers.
\newblock In David~P. Woodruff, editor, {\em Proceedings of the 2024 {ACM-SIAM}
  Symposium on Discrete Algorithms, {SODA} 2024, Alexandria, VA, USA, January
  7-10, 2024}, pages 1568--1605. {SIAM}, 2024.
\newblock \href {https://doi.org/10.1137/1.9781611977912.63}
  {\path{doi:10.1137/1.9781611977912.63}}.

\bibitem{Chuzhoy12}
Julia Chuzhoy.
\newblock On vertex sparsifiers with {Steiner} nodes.
\newblock In Howard~J. Karloff and Toniann Pitassi, editors, {\em Proceedings
  of the 44th Symposium on Theory of Computing Conference, {STOC} 2012, New
  York, NY, USA, May 19 - 22, 2012}, pages 673--688. {ACM}, 2012.
\newblock \href {https://doi.org/10.1145/2213977.2214039}
  {\path{doi:10.1145/2213977.2214039}}.

\bibitem{ConfortiCZ14}
Michele Conforti, Gerard Cornuejols, and Giacomo Zambelli.
\newblock {\em Integer Programming}.
\newblock Springer, 2014.
\newblock \href {https://doi.org/10.1007/978-3-319-11008-0}
  {\path{doi:10.1007/978-3-319-11008-0}}.

\bibitem{Courcelle90/iandc}
Bruno Courcelle.
\newblock The monadic second-order logic of graphs. i. recognizable sets of
  finite graphs.
\newblock {\em Inf. Comput.}, 85(1):12--75, 1990.
\newblock \href {https://doi.org/10.1016/0890-5401(90)90043-H}
  {\path{doi:10.1016/0890-5401(90)90043-H}}.

\bibitem{CyganFKLMPPS15}
Marek Cygan, Fedor~V. Fomin, Lukasz Kowalik, Daniel Lokshtanov, D{\'{a}}niel
  Marx, Marcin Pilipczuk, Michal Pilipczuk, and Saket Saurabh.
\newblock {\em Parameterized Algorithms}.
\newblock Springer, 2015.
\newblock \href {https://doi.org/10.1007/978-3-319-21275-3}
  {\path{doi:10.1007/978-3-319-21275-3}}.

\bibitem{EnglertGKRTT14}
Matthias Englert, Anupam Gupta, Robert Krauthgamer, Harald R{\"{a}}cke, Inbal
  Talgam{-}Cohen, and Kunal Talwar.
\newblock Vertex sparsifiers: New results from old techniques.
\newblock {\em {SIAM} J. Comput.}, 43(4):1239--1262, 2014.
\newblock \href {https://doi.org/10.1137/130908440}
  {\path{doi:10.1137/130908440}}.

\bibitem{FialaGK11}
Jir{\'{\i}} Fiala, Petr~A. Golovach, and Jan Kratochv{\'{\i}}l.
\newblock Parameterized complexity of coloring problems: Treewidth versus
  vertex cover.
\newblock {\em Theor. Comput. Sci.}, 412(23):2513--2523, 2011.
\newblock \href {https://doi.org/10.1016/J.TCS.2010.10.043}
  {\path{doi:10.1016/J.TCS.2010.10.043}}.

\bibitem{FominLMT18}
Fedor~V. Fomin, Mathieu Liedloff, Pedro Montealegre, and Ioan Todinca.
\newblock Algorithms parameterized by vertex cover and modular width, through
  potential maximal cliques.
\newblock {\em Algorithmica}, 80(4):1146--1169, 2018.
\newblock \href {https://doi.org/10.1007/S00453-017-0297-1}
  {\path{doi:10.1007/S00453-017-0297-1}}.

\bibitem{GimaHKKO22}
Tatsuya Gima, Tesshu Hanaka, Masashi Kiyomi, Yasuaki Kobayashi, and Yota
  Otachi.
\newblock Exploring the gap between treedepth and vertex cover through vertex
  integrity.
\newblock {\em Theor. Comput. Sci.}, 918:60--76, 2022.
\newblock \href {https://doi.org/10.1016/J.TCS.2022.03.021}
  {\path{doi:10.1016/J.TCS.2022.03.021}}.

\bibitem{GimaO22}
Tatsuya Gima and Yota Otachi.
\newblock Extended {MSO} model checking via small vertex integrity.
\newblock In Sang~Won Bae and Heejin Park, editors, {\em 33rd International
  Symposium on Algorithms and Computation, {ISAAC} 2022, December 19-21, 2022,
  Seoul, Korea}, volume 248 of {\em LIPIcs}, pages 20:1--20:15. Schloss
  Dagstuhl - Leibniz-Zentrum f{\"{u}}r Informatik, 2022.
\newblock \href {https://doi.org/10.4230/LIPICS.ISAAC.2022.20}
  {\path{doi:10.4230/LIPICS.ISAAC.2022.20}}.

\bibitem{GomoryH61}
R.~E. Gomory and T.~C. Hu.
\newblock Multi-terminal network flows.
\newblock {\em Journal of the Society for Industrial and Applied Mathematics},
  9(4):551--570, 1961.
\newblock URL: \url{http://www.jstor.org/stable/2098881}.

\bibitem{GoranciR16}
Gramoz Goranci and Harald R{\"{a}}cke.
\newblock Vertex sparsification in trees.
\newblock In Klaus Jansen and Monaldo Mastrolilli, editors, {\em Approximation
  and Online Algorithms - 14th International Workshop, {WAOA} 2016, Aarhus,
  Denmark, August 25-26, 2016, Revised Selected Papers}, volume 10138 of {\em
  Lecture Notes in Computer Science}, pages 103--115. Springer, 2016.
\newblock \href {https://doi.org/10.1007/978-3-319-51741-4\_9}
  {\path{doi:10.1007/978-3-319-51741-4\_9}}.

\bibitem{HagerupKNR98}
Torben Hagerup, Jyrki Katajainen, Naomi Nishimura, and Prabhakar Ragde.
\newblock Characterizing multiterminal flow networks and computing flows in
  networks of small treewidth.
\newblock {\em J. Comput. Syst. Sci.}, 57(3):366--375, 1998.
\newblock \href {https://doi.org/10.1006/JCSS.1998.1592}
  {\path{doi:10.1006/JCSS.1998.1592}}.

\bibitem{JambulapatiLLS23}
Arun Jambulapati, James~R. Lee, Yang~P. Liu, and Aaron Sidford.
\newblock Sparsifying sums of norms.
\newblock In {\em 64th {IEEE} Annual Symposium on Foundations of Computer
  Science, {FOCS} 2023, Santa Cruz, CA, USA, November 6-9, 2023}, pages
  1953--1962. {IEEE}, 2023.
\newblock \href {https://doi.org/10.1109/FOCS57990.2023.00119}
  {\path{doi:10.1109/FOCS57990.2023.00119}}.

\bibitem{KhanR14/ipl}
Arindam Khan and Prasad Raghavendra.
\newblock On mimicking networks representing minimum terminal cuts.
\newblock {\em Inf. Process. Lett.}, 114(7):365--371, 2014.
\newblock \href {https://doi.org/10.1016/J.IPL.2014.02.011}
  {\path{doi:10.1016/J.IPL.2014.02.011}}.

\bibitem{Korhonen21}
Tuukka Korhonen.
\newblock A single-exponential time 2-approximation algorithm for treewidth.
\newblock In {\em 62nd {IEEE} Annual Symposium on Foundations of Computer
  Science, {FOCS} 2021, Denver, CO, USA, February 7-10, 2022}, pages 184--192.
  {IEEE}, 2021.
\newblock \href {https://doi.org/10.1109/FOCS52979.2021.00026}
  {\path{doi:10.1109/FOCS52979.2021.00026}}.

\bibitem{KrauthgamerM23}
Robert Krauthgamer and Ron Mosenzon.
\newblock Exact flow sparsification requires unbounded size.
\newblock In Nikhil Bansal and Viswanath Nagarajan, editors, {\em Proceedings
  of the 2023 {ACM-SIAM} Symposium on Discrete Algorithms, {SODA} 2023,
  Florence, Italy, January 22-25, 2023}, pages 2354--2367. {SIAM}, 2023.
\newblock \href {https://doi.org/10.1137/1.9781611977554.CH91}
  {\path{doi:10.1137/1.9781611977554.CH91}}.

\bibitem{KrauthgamerR13}
Robert Krauthgamer and Inbal Rika.
\newblock Mimicking networks and succinct representations of terminal cuts.
\newblock In Sanjeev Khanna, editor, {\em Proceedings of the Twenty-Fourth
  Annual {ACM-SIAM} Symposium on Discrete Algorithms, {SODA} 2013, New Orleans,
  Louisiana, USA, January 6-8, 2013}, pages 1789--1799. {SIAM}, 2013.
\newblock \href {https://doi.org/10.1137/1.9781611973105.128}
  {\path{doi:10.1137/1.9781611973105.128}}.

\bibitem{LampisM21}
Michael Lampis and Valia Mitsou.
\newblock Fine-grained meta-theorems for vertex integrity.
\newblock In Hee{-}Kap Ahn and Kunihiko Sadakane, editors, {\em 32nd
  International Symposium on Algorithms and Computation, {ISAAC} 2021, December
  6-8, 2021, Fukuoka, Japan}, volume 212 of {\em LIPIcs}, pages 34:1--34:15.
  Schloss Dagstuhl - Leibniz-Zentrum f{\"{u}}r Informatik, 2021.
\newblock \href {https://doi.org/10.4230/LIPICS.ISAAC.2021.34}
  {\path{doi:10.4230/LIPICS.ISAAC.2021.34}}.

\bibitem{LeightonM10/stoc}
Frank~Thomson Leighton and Ankur Moitra.
\newblock Extensions and limits to vertex sparsification.
\newblock In Leonard~J. Schulman, editor, {\em Proceedings of the 42nd {ACM}
  Symposium on Theory of Computing, {STOC} 2010, Cambridge, Massachusetts, USA,
  5-8 June 2010}, pages 47--56. {ACM}, 2010.
\newblock \href {https://doi.org/10.1145/1806689.1806698}
  {\path{doi:10.1145/1806689.1806698}}.

\bibitem{MakarychevM10}
Konstantin Makarychev and Yury Makarychev.
\newblock Metric extension operators, vertex sparsifiers and lipschitz
  extendability.
\newblock In {\em 51th Annual {IEEE} Symposium on Foundations of Computer
  Science, {FOCS} 2010, October 23-26, 2010, Las Vegas, Nevada, {USA}}, pages
  255--264. {IEEE} Computer Society, 2010.
\newblock \href {https://doi.org/10.1109/FOCS.2010.31}
  {\path{doi:10.1109/FOCS.2010.31}}.

\bibitem{Moitra09}
Ankur Moitra.
\newblock Approximation algorithms for multicommodity-type problems with
  guarantees independent of the graph size.
\newblock In {\em 50th Annual {IEEE} Symposium on Foundations of Computer
  Science, {FOCS} 2009, October 25-27, 2009, Atlanta, Georgia, {USA}}, pages
  3--12. {IEEE} Computer Society, 2009.
\newblock \href {https://doi.org/10.1109/FOCS.2009.28}
  {\path{doi:10.1109/FOCS.2009.28}}.

\bibitem{NagamochiI92}
Hiroshi Nagamochi and Toshihide Ibaraki.
\newblock A linear-time algorithm for finding a sparse k-connected spanning
  subgraph of a k-connected graph.
\newblock {\em Algorithmica}, 7(5{\&}6):583--596, 1992.
\newblock \href {https://doi.org/10.1007/BF01758778}
  {\path{doi:10.1007/BF01758778}}.

\bibitem{OostveenL23}
Jelle~J. Oostveen and Erik~Jan van Leeuwen.
\newblock Streaming deletion problems parameterized by vertex cover.
\newblock {\em Theor. Comput. Sci.}, 979:114178, 2023.
\newblock \href {https://doi.org/10.1016/J.TCS.2023.114178}
  {\path{doi:10.1016/J.TCS.2023.114178}}.

\bibitem{RajagopalanV99}
Sridhar Rajagopalan and Vijay~V. Vazirani.
\newblock On the bidirected cut relaxation for the metric {Steiner} tree
  problem.
\newblock In Robert~Endre Tarjan and Tandy~J. Warnow, editors, {\em Proceedings
  of the Tenth Annual {ACM-SIAM} Symposium on Discrete Algorithms, 17-19
  January 1999, Baltimore, Maryland, {USA}}, pages 742--751. {ACM/SIAM}, 1999.
\newblock URL: \url{http://dl.acm.org/citation.cfm?id=314500.314909}.

\bibitem{SpielmanT04/stoc}
Daniel~A. Spielman and Shang{-}Hua Teng.
\newblock Nearly-linear time algorithms for graph partitioning, graph
  sparsification, and solving linear systems.
\newblock In L{\'{a}}szl{\'{o}} Babai, editor, {\em Proceedings of the 36th
  Annual {ACM} Symposium on Theory of Computing, Chicago, IL, USA, June 13-16,
  2004}, pages 81--90. {ACM}, 2004.
\newblock \href {https://doi.org/10.1145/1007352.1007372}
  {\path{doi:10.1145/1007352.1007372}}.

\end{thebibliography}

\clearpage
\appendix

\section{Remaining Proofs}

\lemBasicTransitivity*
\begin{proof}
Let $A \subseteq K'$. By the definition of cut sparsifiers, we have that:
\begin{alignat*}{2}
\mc_G(A) &\leq \mc_H(A) &&\leq q\cdot\mc_G(A), \\
\mc_H(A) &\leq \mc_L(A) &&\leq r\cdot\mc_H(A).
\end{alignat*}
Putting both inequalities together, this implies that
\begin{alignat*}{2}
\mc_G(A) &\leq \mc_L(A) &&\leq qr\cdot\mc_G(A).
\end{alignat*}

Similarly, for flow sparsifiers, we know that for every demand $\bfd$ for
terminals $K'$, $\lambda_L(\bfd) \leq \lambda_H(\bfd) \leq \lambda_G(\bfd)$. %
On the other hand, $\lambda_G(\bfd) \leq q \cdot \lambda_H(\bfd) \leq qr \cdot
\lambda_L(\bfd)$, thus completing the proof.
\end{proof}

\lemBasicUplus*
\begin{proof}
We assume that $K = V(G_1) \cap V(G_2)$. %
Otherwise, for the purpose of the analysis, we can add to $G_1$ and $H_1$ the
terminals in $K \setminus V(G_1)$ as isolated vertices (and the ones in $K
\setminus V(G_2)$ to $G_2$ and $H_2$). %
This preserves the properties of the sparsifiers $H_1$, $H_2$ and does not
affect the graphs $G$ or $H$, and thus the lemma holds even without the assumption.

Let $A \subseteq K$ and $X = \mc_G(A)$. %
Then, $u_G(X) = u_{G_1}(X \cap V(G_1)) + u_{G_2}(X \cap
V(G_2))$, as every edge with a Steiner endpoint appears only in one of the
terms, and edges appearing in both terms satisfy $u_G(e) = u_{G_1}(e) +
u_{G_2}(e)$. %
Using the properties of $H_1$ and $H_2$, we get that for $Y_1 = \mc_{H_1}(A)$, $Y_2=\mc_{H_2}(A)$,
\begin{align*}
\kappa_H(A)
&\leq u_H(Y_1 \cup Y_2) \\
&= u_{H_1}(Y_1) + u_{H_2}(Y_2) \\
&\leq u_{G_1}(X \cap V(G_1)) + u_{H_2}(X \cap V(G_2)) \\
&= u_{G}(X),
\end{align*}
and thus $\kappa_H(A) \leq \kappa_G(A)$.

For the other direction, we claim that $X \cap V(G_1)$ and $X \cap V(G_2)$
have to be optimal cuts for $A$ in $G_1$, $G_2$, respectively. %
Note that $V(G_1) \cap V(G_2) = K$, and so $X \cap V(G_1) \cap V(G_2) = X \cap K = A$. %
Thus, if (w.l.o.g.) $X \cap V(G_1)$ were not optimal for $G_1$, we could replace
$X \cap (V(G_1) \setminus K)$ with a better cut, and thus obtain a better cut $X$
for $G$. %

We conclude that
\begin{align*}
\kappa_G(A)
&= u_{G_1}(X \cap V(G_1)) + u_{G_2}(X \cap V(G_2)) \\
&\leq q \cdot \kappa_{H_1}(A) + q \cdot \kappa_{H_2}(A) \\
&\leq q \cdot u_{H_1}(Y \cap V(G_1)) + q \cdot u_{H_2}(Y \cap V(G_2)) \\
&= q \cdot u_{H}(Y),
\end{align*}
where $Y$ is a min-cut for $A$ in $H$.
\end{proof}

\lemBasicEquivMerge*
\begin{proof}
Let $A \subseteq K$ be any subset of terminals. %
We need to show that $\kappa_G(A) = \kappa_{G/vw}(A)$. %

To start, notice that any cut in $G / vw$ induces a cut on $G$ with the same
capacity (where if the combined vertex $vw$ is in the cut, we include both $v$
and $w$ in $G$). %
Thus, $\kappa_G(A) \leq \kappa_{G/vw}(A)$, as the min-cut in $G/vw$ induces a
cut in $G$.

Consider now a min-cut separating $A$ in $G$. We will show that $G/vw$ has a
cut of the same capacity, thus showing that $\kappa_G(A) \geq \kappa_{G/vw}(A)$. %
By assumption, there is a min-cut $X$ separating $A$ that either contains both
$v$ and $w$ or $X \cap \set{u,v} = \emptyset$. %
In the first case, we take $X'=X - \set{v,w} + \set{vw}$ as the cut in $G/vw$,
and in the second case we simply take $X'=X$. %
Either way, $u_{G/vw}(X') = u(X)$, completing the proof.
\end{proof}

\lemViSparsifier*
\begin{proof}
We recall the definitions of $\Ccal_Q(G_S) = \set{C'_q: q \in Q}$ and $C'_q =
C_q \cdot \sum_{C \in \Ccal} \lambda_C(q)$. %
We already know that $|Q| \leq 4^{b(a+b)^2}$, and thus $H$ has size at most $4^{b(a+b)^2}$.

The main tool for the proof is the following lemma, whose proof is similar to \Cref{lem:basic:equiv-merge}.

\begin{restatable}{lemma}{lemViMerge}
\label{lem:vi:merge}
Let $C_v$, $C_w$ be components of $G-X$ of size $b$ with capacity vectors $c_v$, $c_w$, respectively, and vertices $V(C_v) = \set{v_1,
\ldots, v_b}$, $V(C_w) = \set{w_1, \ldots, w_b}$.

If $C_v$ agrees with $C_w$, i.e.~$\Scal_{C_v} \supseteq \Scal_{C_w}$, then the
graph $G / (C_v,C_w)$, obtained by contracting each pair of vertices $v_iw_i$,
is a sparsifier for $G$.
\end{restatable}

This lemma allows us to show that any component can be decomposed into a conic
combination of basic components:

\begin{restatable}{lemma}{lemViDecomp}
\label{lem:vi:decomp}
Let $C$ be a component of $G-X$ of size $b$ and capacity vector $c$ for its edges. %

If  $c$ can be written as the conic combination of points $q_1, q_2, \ldots,
{q_{\ell} \in Q}$, all of which agree with $c$, %
then $G$ has a sparsifier given by $(G - C) \uplus C'_1 \uplus \ldots \uplus
C'_{\ell}$, where each $C'_i$ is a graph on $X \cup \set{v'_{i,1}, \ldots,
v'_{i,b}}$, and the capacities of the edges in $C'_i$ are given by $\lambda(q_i) \cdot q_i$.
\end{restatable}

Finally, using \Cref{lem:vi:merge}, we can contract the different copies of
basic components we created, to obtain the sparsifier for $G$.
\end{proof}

\end{document}